\documentclass[12pt, draftclsnofoot, onecolumn]{IEEEtran}
\usepackage{amssymb}
\usepackage{graphicx}
\usepackage{amsmath}
\usepackage{epsf}
\usepackage{mathrsfs}
\usepackage{cite}
\newtheorem{theorem}{Theorem}
\newtheorem{corollary}{Corollary} 
\newtheorem{lemma}{Lemma}
\newtheorem{remark}{Remark} 

\newtheorem{example}{Example}  

\begin{document}
\title{{Channel Aided Interference Alignment}}
\author{\Large Zainalabedin ~Samadi,  
       ~Vahid ~Tabatabavakili and ~Farzan ~Haddadi
\\\small Dept. of Elec. Eng.,   Iran University of Sceince and Technology 
Tehran,   Iran 
\\ \{z.samadi\}@elec.iust.ac.ir
\\ \{vakily,   haddadi\}@iust.ac.ir
}
\maketitle
\begin{abstract}

Interference alignment (IA) techniques mostly attain their degrees of freedom (DoF) benefits as the number of channel extensions tends to infinity. Intuitively, the more interfering signals that need to be aligned, the larger the number of dimensions needed to align them. This requirement poses a major challenge for IA in practical systems. This work evaluates the necessary and sufficient conditions on channel structure of a fully connected interference network with time-varying fading to make perfect IA feasible within limited number of channel extensions. We propose a method based on the obtained conditions on the channel structure to achieve perfect IA. For the case of $3$ user interference channel,   it is shown that only one condition on channel coefficients is required to make perfect IA feasible at all receivers. IA feasibility literature have mainly focused on network topology  so far. In contrast, derived channel aiding conditions in this work can be considered as the perfect IA feasibility conditions on channel structure.
\end{abstract}
\IEEEdisplaynotcompsoctitleabstractindextext
\section{Introduction}
Wireless network receivers should  cope with interference from undesired transmitters in addition to the ambient noise,  and hence,  there is a rising interest in using advanced interference mitigation techniques to improve the network performance.  IA is One of the latest strategies to deal with interference. This technique was first introduced by Maddah Ali et. al. \cite{Maddah08}.  The idea of interference alignment (IA) is to coordinate multiple transmitters so that their mutual interference aligns at the receivers, facilitating simple interference cancellation techniques. The original method proposed in \cite{Maddah08} was iterative and nonlinear,  the linear and closed form approach was introduced by Jafar and Shamai \cite{Jafar}.

The majority of IA schemes fall into one of two categories  of signal space alignment and signal level alignment. Our main focus in this paper is on signal space alignment schemes. Using infinite dimensional extension of the channel,  it is shown that IA could achieve the optimal DoF of $K/2$ in K-pair single antenna ergodic interference channels  \cite{Cadam08}. The idea of IA has been successfully applied to various interference networks.

 Ergodic IA (EIA) scheme is proposed by Nazer et al., \cite{Nazer12}. This scheme aims to achieve $1/2$ interference-free ergodic capacity of interference channel (IFC) at any signal-to-noise ratio. The order of channel extensions needed by  \cite{Nazer12} is roughly the same as \cite{Cadam08}. The similar idea of opportunistically pairing two channel instances to cancel interference has been proposed  independently by \cite{Sang13} as well.  However,   EIA scheme  is based on an special pairing of the channel matrices and does not address the general structure of the paired channels  suitable for cancelling interference.

Assuming linear combining of paired channel output signals, this paper addresses  the general necessary and sufficient  structure of channel matrices which are suitable to be paired to cancel interference. Using this general pairing scheme, a new IA scheme is proposed which  significantly reduces the number of channel extensions needed to perfectly align interference. 

From a different standpoint, this paper obtains the necessary and sufficient feasibility conditions on channel structure to achieve total DoF of the IFC using limited number of channel extension. So far, IA feasibility literature have mainly focused on network configuration, see \cite{Ruan} and references therein. To ease some of IA criteria by using channel structure,  \cite{Leejan09} investigates degrees of freedom for the partially connected IFCs where some arbitrary interfering links are assumed disconnected. In this channel model,   \cite{Leejan09} examines how these disconnected links are considered on designing the beamforming vectors for IA and closed-form solutions are obtained for some specific configurations. In contrast,  our work evaluates the necessary and sufficient conditions on channel structure of an IFC to make perfect IA possible with limited number of channel extensions. An earlier version of this paper have been reported in \cite{Samadi}.

This paper is organized as follows. The system model is given in Section $2$.  In section $3$,  it is argued why LIA cannot achieve optimal degrees of freedom with limited number of channel extensions. We present the proposed scheme in section $4$. An example of the proposed method application is presented in section $5$.   In section $6$, the required number of channel extensions or equivalently the delay incurred by the proposed scheme is discussed and is shown to be significantly lower than previous proposed schemes. Section $6$ concludes our paper.

\section{System Model} \label{secsysmod}

\begin{figure}
\centering \includegraphics[scale=1.2]{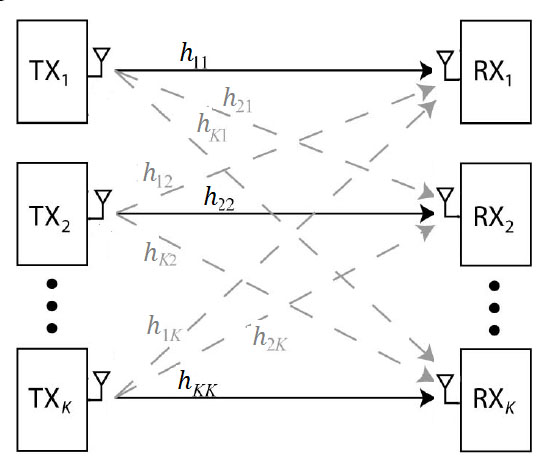}
\caption{K user IFC Model.}
\label{figure:KUser}
\end{figure}

Consider a $K$ user single-hop single antenna interference network. An illustration of system model is shown  in Fig. \ref{figure:KUser}. Each transmitter wishes to communicate with its respective receiver. Communication takes place in a shared bandwidth and the goal is to achieve maximum possible  sum rate along with a reliable communication. 

The channel between transmitter $j$ and receiver $k$, $j, k\in\{1,  \ldots,  K\}$ at time instant $t\in \mathbb{N}$ is denoted as  $h^{[kj]} (t)$. We assume that the values of channel coefficients at different time instants are independently drawn from some continuous distribution. The channel gains are bounded between a positive minimum value and a finite maximum value to avoid degenerate channel conditions.  The channel output observed by receiver  $k\in\{1,  \ldots,  K\}$  at time slot $t\in \mathbb{N}$ is a noisy linear combination of the inputs
\begin{eqnarray}
y^{[k]}(t)=h^{[k1]}(t)x^{[1]}(t)+h^{[k2]}(t)x^{[2]}(t) \cdots \nonumber \\+h^{[kK]}(t)x^{[K]}(t)+z^{[k]}(t), 
\end{eqnarray}
where $x^{[k]} (t)$  is the transmitted signal  of the $k^{\textrm{th}}$ transmitter,  and $z^{[k]}(t)$  is additive independent and identically distributed noise and drawn from a circularly
symmetric complex Gaussian distribution with unit variance, $z^{[k]}(t) \sim \mathcal{CN}(0, 1)$. It is assumed that all transmitters are subject to the power constraint:

\begin{eqnarray}
\mathrm{E}(||z^{[k]}(t)||^{2}) \leq P, \quad 1 \leq k \leq K.
\end{eqnarray}

We assume that there is perfect channel state information (CSI) at receivers and global CSI at transmitters. Hereafter,   time index is omitted for the sake of simplicity.

\section{Linear Vector IA Limitation}

Degrees-of-freedom region for a $K$ user IFC, with the system model discussed in section \ref{secsysmod}, has been derived in \cite{Cadam08} as follows, 
\begin{eqnarray}
\mathcal{D}= \left \{ {\bf d} \in \mathbb{R}_+^K: d_i+ d_j \leq 1, \; 1 \leq i, j\leq K \right \}, 
\label{dofreg}
\end{eqnarray}
and the number of DoF achieved by $K$ user IFC is obtained to be $K/2$. It is straightforward to see that the following corollary describes the only DoF vector, ${\bf d}$, that achieves total number of DoF. 

\begin{corollary}
\label{cor1}
The only DoF vector that achieves   total number of DoF of an IFC  is
\begin{eqnarray}
d_i=\frac{1}{2}, \forall 1\leq i \leq K.
\label{optdof}
\end{eqnarray}
\end{corollary}

Consider a $3$ user IFC. We will use the scheme based on \cite{Cadam08} to do IA. Assuming channel coefficients to be generic, it has been shown in \cite{Cadam08}  that optimal total of DoF for a $3$ user IFC cannot be achieved over limited number of channel usage. A brief review is presented here to maintain continuity of presentation. 

Let $\tau$ denote the duration of the time expansion in number of symbols. Here and after, we use the
upper case bold font to denote the time-expanded signals, e.g., ${\bf H}^{[jk]} = \textrm{diag} (h^{[jk]}(1) \ldots,h^{[jk]}(\tau))$, which is a size $\tau \times \tau$ diagonal matrix. Denote the beamforming matrix of transmitter $k$ as ${\bf V}^{[k]}$. 

We intend to achieve the outer bound of  $3/2$ DoF for this setup. Consider $2n$ extension of the channel. Over this extended channel,  consider a hypothetical achievable scheme where each of transmitter achieves $n$ DoF if possible,  using beamforming at every transmitter and zero-forcing at every receiver. Note that this is the only DoF point in achievable region that achieves total number of DoF of this network, according to corollary \ref{optdof}. The signal vector at the $k$'th receiver can be stated as
\begin{eqnarray}
{\bf Y}^{[k]}&{}={}&{\bf H}^{[k1]}{\bf V}^{[1]}  {\bf X}^{[1] }+{\bf H}^{[k2]}{\bf V}^{[2]}  {\bf X}^{[2] }\nonumber \\ &&{+}\: {\bf H}^{[k3]}{\bf V}^{[3]}  {\bf X}^{[3] },+{\bf Z}^{[k]},
\end{eqnarray}
where ${\bf H}^{[k1]}$ is the $2n \times 2n$ extension of the channel, ${\bf V}^{[k]}$ is $2n \times n$   beamforming matrix of user $k$, and ${\bf X}^{[k]}$ is a $n\times 1$ column vector comprised of transmitted symbols  $x_m^{[k]},   m=1,  \ldots,   n $. ${\bf Y}^{[k]}$ and ${\bf Z}^{[k]}$  represent the $2n$ symbol extension of $y^{[k]}$ and $z^{[k]}$,  respectively.

Receiver $i$ cancels the interference by  zero forcing all ${\bf V}^{[j]}, j\neq i$. The vectors corresponding to  interfering vectors must not occupy more than $n$ dimensions  from the $2n$ dimensional received signal vector ${\bf Y}^{[k]}$. Thus, IA requirements can be written as follows:

\begin{eqnarray}
\textrm{span} \left ( {\bf H}^{[ik]} {\bf V}^{[k]}  \right )=\textrm{span} \left ( {\bf H}^{[ij]} {\bf V}^{[j]} \right ), \quad i, j, k=1, 2, 3, \quad k, j \neq i, 
\label{SE1}        
\end{eqnarray}
where $\textrm{span}({\bf A})$ denotes the column space of matrix ${\bf A}$. 

Note that the channel matrices $ {\bf H}^{[ji]}$ are full rank almost surely. Using this fact,   (\ref{SE1}) implies that 
\begin{eqnarray}
\textrm{span} \left (  {\bf V}^{[1]}  \right )=\textrm{span} \left ( {\bf T} {\bf V}^{[1]} \right ), 
\label{CAE1}           
\end{eqnarray}
where $ {\bf T}$ is defined as 
\begin{eqnarray}
 {\bf T}= ( {\bf H}^{[13]}  ) ^{-1} {\bf H}^{[23]}   ( {\bf H}^{[21]})^{-1} {\bf H}^{[12]} ( {\bf H}^{[32]}  ) ^{-1}  {\bf H}^{[31]}.
\label{TM}
\end{eqnarray} 

 (\ref{CAE1}) implies that at least one eigenvector of  ${\bf T}$ is in  $\textrm{span} \left (  {\bf V}^{[1]}  \right )$. Since all channel matrices are diagonal,   the set of eigenvectors for all channel matrices,   their inverse and product are all identical to the set of column vectors of the identity matrix,  namely vectors of the from ${\bf e}_k=[0 \; 0 \; \cdots \; 1 \; \cdots \; 0]^T$. Since ${\bf e}_k$ exists in $\textrm{span} \left (  {\bf V}^{[1]}  \right )$,   (\ref{SE1}) implies that 
\begin{eqnarray}
&{}& {\bf e}_k \in \textrm{span} \left ( {\bf H}^{[ij]} {\bf V}^{[j]}  \right ),   \quad \forall i,   j \in \{1,  2,  3\}.         
\label{imply}
\end{eqnarray}
(\ref{imply}) implies that  at receiver $1$,   the desired signal $ {\bf H}^{[11]} {\bf V}^{[1]} $  is not linearly independent of the interference signal,   ${\bf H}^{[12]} {\bf V}^{[2]}$,   and therefore,   receiver $1$ cannot fully recover $ {\bf X}^{[1]}$ only by zero forcing the interference signal. Therefore, $3/2$ degrees of freedom for the $3$ user single antenna IFC cannot be achieved through LIA schemes, assuming channel coefficients to be completely random and generic.

\section{Channel Aided IA (CAIA)}
 The main result of this paper is  summarized in the following theorem:
\begin{theorem}
\label{maintheo}
The necessary and sufficient conditions for the perfect IA to be feasible in a $K$ user IFC is to have the following structure on the channel matrices:
\begin{eqnarray}
 {\bf T}_{j}^{[i]}={\bf P} \left [ \begin{array}{c c c} \tilde{{\bf T}}_{j}^{[i]} & 0 & 0 \\ 0 & \tilde{{\bf T}}_{j}^{[i]}& 0 \\ 0 & 0 & f(\tilde{{\bf T}}_{j}^{[i]}) \end{array} \right ] {\bf P}^T,  \\  \quad  i,  j \in \{2,   \ldots,   K\},   i \not = j, 
 \label{KAJI}
 \end{eqnarray}
 where ${\bf T}_j^{[i]}$ matrices are defined as 
\begin{eqnarray}
\begin{split}
{\bf T}_j^{[i]}=&\left ( {\bf H}^{[i1]} \right )^{-1} {\bf H}^{[ij]} \left ( {\bf H}^{[1j]} \right )^{-1}\\& {\bf H}^{[13]} \left ( {\bf H}^{[23]} \right )^{-1}{\bf H}^{[21]}, \\  \quad & i,  j =2, \ldots,  K,   \quad  j\not = i. \end{split}
\end{eqnarray}
 ${\bf P}$ is a $2n \times 2n$ permutation matrix,  and  $\tilde{{\bf T}}_j^{[i]}$ is an arbitrary $n_1 \times n_1$ diagonal matrix with nonzero diagonal elements, 	with an arbitrary $n_1$ in the range $1 \leq n_1 \leq n$,  and  $f({\bf X})$ is a  mapping whose domain is an $n_1 \times n_1$ diagonal  matrix and range is a $2n-2n_1 \times 2n-2n_1$  diagonal matrix ${\bf Y}=f({\bf X})$ whose set of diagonal elements  is a subset of diagonal elements of  ${\bf X}$. Diagonal elements of $ f(\tilde{{\bf T}}_{j}^{[i]})$ are distinct. 
\end{theorem}

\begin{proof}

The proof for the case of $3$ user IFC is presented here, the general case of $K$ user IFC is similar. Based on  Theorem \ref{maintheo}, the necessary and sufficient condition for the perfect IA to be feasible in a $3$ user IFC is to have the following structure on the channel matrices:

\begin{eqnarray}
 &{}&{\bf T}={\bf P} \left [ \begin{array}{c c c} \tilde{{\bf T}} & 0 & 0 \\ 0 & \tilde{{\bf T}}& 0 \\ 0 & 0 & f(\tilde{{\bf T}}) \end{array} \right ] {\bf P}^T,
\label{CAEm}
\end{eqnarray}
where ${\bf T}$ is defined in (\ref{TM}). The proof of the necessary part is perented fist. 

\begin{lemma}
\label{lemma1}
If $2n \times n$ matrix $ {\bf V}^{[1]}$  is full column rank,   (\ref{CAE1}) implies that   $n$  eigenvectors of ${\bf T}$ lie in $\textrm{span} \left (  {\bf V}^{[1]}  \right )$.
\end{lemma}
\begin{proof}
 (\ref{CAE1}) implies that there exists an $n \times n$ dimensional matrix ${\bf A}$ such that 
\begin{eqnarray}
{\bf T} {\bf V}^{[1]}={\bf V}^{[1]}{\bf A}.
\end{eqnarray}
Let's define ${\bf w}$ as an eigenvector of ${\bf A}$, i.e., ${\bf A}{\bf w}=\mu{\bf w}$ where $\mu$ is its corresponding eigenvalue.  Since ${\bf V}^{[1]}$ is full column rank,  ${\bf V}^{[1]} {\bf w} \not = 0$ and we can write:
\begin{eqnarray}
{\bf T} {\bf V}^{[1]}{\bf w}={\bf V}^{[1]}{\bf A}{\bf w}=\mu {\bf V}^{[1]}{\bf w}.
\end{eqnarray}
Then ${\bf V}^{[1]}{\bf w}$ is an eigenvector of ${\bf T}$. On the other hand, ${\bf V}^{[1]}{\bf w}$ is in $\textrm{span} \left ( {\bf V}^{[1]}  \right )$.   Since  ${\bf A}$ has $n$ independent eigenvectors,   therefore,  $n$  of eigenvectors of ${\bf T}$ lie within $\textrm{span} \left ( {\bf V}^{[1]}  \right )$.
\end{proof}

On the other hand, based on the discussion we had on  (\ref{CAE1}),  there should not be any vector of the form ${\bf e}_i$ such that  ${\bf e}_i \in \textrm{span} \left ( {\bf V}^{[1]}  \right )$.  Since $\textrm{span} \left ( {\bf V}^{[1]}  \right )$ has dimension $n$,   it should have $n$ basis vectors of the form $ {\bf v}=\sum_{i=1}^{2n}\alpha_i {\bf e}_i,   \quad j=1, \ldots,  n$,   where at least $2$ of $\alpha_i$'s are nonzero. Let's call vectors with this form as non ${\bf e}_i$ vectors. Since $n$ of  ${\bf T}$'s eigenvectors lie in $\textrm{span} \left ( {\bf V}^{[1]}  \right )$,   the matrix ${\bf T}$ should have at least $n$ non ${\bf e}_i $ eigenvectors. Note that this requirement is necessary not sufficient. Assuming that ${\bf S}=[{\bf s}]$ is a matrix consisted of non ${\bf e}_i $ eigenvectors of ${\bf T}$ as its columns, it is concluded that $\textrm{span} \left ( {\bf V}^{[1]}  \right ) \in \textrm{span} \left ( {\bf S}  \right )$.

\begin{lemma}
\label{lemma2}
${\bf T}$ has no unique diagonal element.
\end{lemma}
\begin{proof}
It can easily be shown that if ${\bf s}_1= {\bf e}_i + {\bf e}_j,   \quad i,j=1,    \ldots,   n, i \neq j$  is an eigenvector of ${\bf T}$,   then $t_i = t_j$.   If $t_l$ is unique,  this implies that non ${\bf e}_i $ eigenvectors of ${\bf T}$ do not contain ${\bf e}_l$, and hence,  ${\bf e}_l \in \textrm{kernel} \left ( {\bf S}  \right )$, where $ \textrm{kernel} \left ( {\bf S}  \right )$  denotes  the null space of columns of matrix ${\bf S}$. Thus, $ {\bf e}_l \in \textrm{kernel} \left ({\bf V}^{[1]} \right )$ because  $\textrm{span} \left ( {\bf V}^{[1]}  \right ) \in \textrm{span} \left ( {\bf S}  \right )$. Since all channel matrices are diagonal,  using (\ref{SE1}),  ${\bf e}_j \in \textrm{kernel}({\bf V}^{[1]})$ implies that 
 \begin{eqnarray}
&{}& {\bf e}_j \in \textrm{kernel} \left ( {\bf H}^{[ij]} {\bf V}^{[j]}  \right ),    \quad \forall i, j \in \{1,   2,   3\}.
\end{eqnarray}

Thus,    at receiver $1$,    the total dimension of the desired signal $ {\bf H}^{[11]} {\bf V}^{[1]}$ plus interference from undesired transmitters, ${\bf H}^{[1j]} {\bf V}^{[j]},  j \neq 1$, is less than $2n$,  and desired signals are not linearly independent from the interference signals,    and hence,    receiver $1$ cannot fully recover  $ {\bf X}^{[1]}$ solely by zeroforcing the interference signal. Lemma \ref{lemma2} concludes the proof of necessary part of (\ref{CAEm}).  
\end{proof}
 
   The sufficient part is proved by noting the fact that  matrix ${\bf T} $ with the form given in (\ref{CAEm}) has  $L \geq n$ non ${\bf e}_i $ eigenvectors ${\bf r}_i, i=1, \ldots, L$ with the property that 
 \begin{eqnarray}
{\bf e}_k \not \in \textrm{span}({\bf R}), \quad k=1, \ldots, 2n,
\label{spnprp}
\end{eqnarray}
 and
 \begin{eqnarray}
  {\bf e}_k \not \in \textrm{kernell}({\bf R}), \quad k=1, \ldots, 2n,
  \label{krnlprp}
  \end{eqnarray}
where ${\bf R} $  is defined as a $2n \times L$ matrix consisted of ${\bf r}_i$'s as its columns.  In fact, one example case of matrix ${\bf R}$ can be obtained as follows: 
\begin{eqnarray} \begin{split}
 &{\bf R}={\bf P} \left [ \begin{array}{c c} \tilde{{\bf V}}&{\bf B}\\ \tilde{{\bf V}} &-{\bf B}\\ 0  & f(\tilde{{\bf V}}) \end{array} \right ], \end{split}
 \label{6by3cacbf}
 \end{eqnarray}
where $\tilde{{\bf V}}$ is an arbitrary $n_1 \times n_1$ diagonal matrix with $n_1$ defined in Theorem \ref{maintheo}. ${\bf P}$ and  f({\bf X})  are the same permutation matrix and mapping function used in (\ref{CAEm}), and ${\bf B}$ is defined as an $n_1 \times (2n-2n_1)$ matrix consisted of $0$ or $1$ elements. It is obtained as follows;
\begin{eqnarray}
 f(\tilde{{\bf T}}_{j}^{[i]}) ={\bf B}^T \tilde{{\bf T}}_{j}^{[i]} {\bf B}.
 \end{eqnarray} 
 The fact that diagonal elements of $ f(\tilde{{\bf T}}_{j}^{[i]})$ are assumed to be distinct implies that matrix  ${\bf B}$ has only one nenzero element in each row and column, therefore, ${\bf B} {\bf B}^T$ is a diagonal matrix, and  ${\bf B}^T {\bf B}={\bf I}_{2n-2n_1}$. Therefore, it is obtained that ${\bf B}f(\tilde{{\bf T}})=\tilde{{\bf T}}{\bf B}$, and hence,  matrix ${\bf R}$ satisfies the following condition, 
\begin{eqnarray}
{\bf T}{\bf R} ={\bf R}\left [ \begin{array}{c c} \tilde{{\bf T}} & 0  \\  0 & f(\tilde{{\bf T}}) \end{array} \right ],
 \end{eqnarray}
which implies that the columns of matrix ${\bf R}$ are eigenvectors of matrix ${\bf T}$. It can also be verified that matrix ${\bf R}$ satisfies (\ref{spnprp}) and  (\ref{krnlprp}).  Note that this choice for the set of non ${\bf e}_i$ eigenvectors of ${\bf T}$ is not unique. The last $n$ columns of matrix  ${\bf R}$ can be considered as  the columns of user $1$ transmit beamforming matrix ${\bf V}^{[1]}$. In this case,   ${\bf V}^{[1]}$ satifies  (\ref{spnprp}) and  (\ref{krnlprp}) as well.  ${\bf V}^{[2]}$   and ${\bf V}^{[3]}$ can be designed using (\ref{SE1}).
\end{proof}
\begin{example}
As an example, assume that, using $6$ extension of the channel,  $6 \times 6$  diagonal matrix  ${\bf T}$ has the following form, 
\begin{eqnarray}
{\bf T}=\textrm{diag}(1, 2, 1, 2, 1, 2)
\label{Texpl}
\end{eqnarray}
which has the form given in (\ref{CAEm}) with $n=3, n_1=2$, ${\bf P}={\bf I}_6$, where  ${\bf I}_6$ is $6 \times 6$ identity matrix,  $\tilde{{\bf T}}=\textrm{diag}(1, 2)$, and  $f(\tilde{{\bf T}})=\tilde{{\bf T}}$. Therefore,  matrix  ${\bf B}$ is obtained as ${\bf B}={\bf I}_2$. Matrix ${\bf R}$ for this example case can be obtained as 
\begin{eqnarray}
{\bf R}=\left [  \begin{array}{c c c c} 1  & 0 & 1  & 0 \\ 0  & 1& 0  & 1 \\ 
1  & 0 & -1  & 0 \\ 0  & 1 &0  &-1 \\ 0  & 0 & 1  & 0 \\ 0  & 0 & 0  & 1 \end{array} \right ].
\end{eqnarray}
The last $3$ colmuns of matrix ${\bf R}$ or every  span of these columns can be considered as the user $1$ transmit beamforming matrix,  ${\bf V}^{[1]}$. ${\bf V}^{[2]}$   and ${\bf V}^{[3]}$ can be obtained using (\ref{SE1}). 
 \end{example}
 
 \begin{remark}
It should be noted in (\ref{KAJI}) that all $(K-1)(K-2)-1$  channel aiding conditions share the same permutation matrix ${\bf P}$ and mapping function $f({\bf X})$. This is because the diagonal matrices  ${\bf T}_j^{[i]}$ should have the same set of non ${\bf e}_i$ eigenvectors that satisfies   (\ref{spnprp}) and  (\ref{krnlprp}), which are supposed to be columns of user $1$ beamforming matrix, ${\bf V}^{[1]}$. This requirement implies that complementary channel coefficients should occur simultaneously.
\end{remark}

\begin{remark}
If the condition (\ref{CAEm}) is true with the following form 
\begin{eqnarray}
 &{}&{\bf T}={\bf P} \left [ \begin{array}{c c} \tilde{{\bf T}} & 0 \\ 0 & \tilde{{\bf T}} \end{array} \right ] {\bf P}^T,
\label{CAEms}
\end{eqnarray}
where $\tilde{{\bf T}}$ is an arbitrary $n \times n$ diagonal matrix, ${\bf V}^{[1]}$   can be designed as 
\begin{eqnarray}
{\bf V}^{[1]}={\bf P}^T \left[ \begin{array}{c} {\bf I}_n \\  {\bf I}_n \end{array} \right ], 
\label{vdesig}
\end{eqnarray}
where ${\bf P}$ is the same permutation matrix used in  (\ref{CAEms}) and ${\bf I}_n$ is the $n \times n$ identity matrix. ${\bf V}^{[1]}$ can also be designed as any other $2n \times n$ matrix having the same column vector subspace with (\ref{vdesig}). 
\end{remark}

\begin{remark}
There are  $(K-1)(K-2)$ equations in (\ref{CAEms}). However,  channel aiding condition is already satisfied for the case of  $j=3$ and $ i=2$, because ${\bf T}_3^{[2]} ={\bf I}_{2n}$. Therefore,  there are $(K-1)(K-2)-1$  perfect IA feasibility conditions on channel structure for general $K$ user IFC.
\end{remark}
\begin{remark}
Let's denote every matrix ${\bf U}$ with the form given in (\ref{CAEm}) as ${\bf U}={\bf T}_P$. It can easily be seen that if  ${\bf U}={\bf T}_P$ and  ${\bf V}={\bf T}_P$, so is  ${\bf U}^{-1}={\bf T}_P$  and  ${\bf U} {\bf V}={\bf T}_P$. Channel aiding conditions have the following form, 
\begin{eqnarray}
 {\bf T}_j^{[i]}= {\bf T}_P,   \quad  i,  j=\{2,  \ldots,   K\},   i \not = j.
\end{eqnarray}
Assume that these conditions have the form given in (\ref{CAEms}). The special case of ${\bf H}^{[ij]} ={\bf T}_P,   \quad \forall i,   j,   \quad i \not = j$ is the channel coefficient pairing algorithm used in EIA in \cite{Nazer12}. In this case, channel aiding conditions, ${\bf T}_j^{[i]}={\bf T}_P$, are already satisfied, and hence, ergodic IA conditions are the special case  of the channel aiding conditions obtained in this paper.
\end{remark}

\section{An Example of Proposed Method Application}

 Consider a $3$ user IFC where transmitters $1,   2\; \textrm{and}\; 3$ are sending  symbols $x_1,   x_2,   \textrm{and} \; x_3$, respectively,  over time slot $t$. The following signal vector is received at the respective receivers:
\begin{eqnarray}
\left[ \begin{array}{c} y_1 \\ y_2 \\ y_3 \end{array} \right]= \left[ \begin{array}{c c c} h_{11} & h_{12} & h_{13} \\ h_{21} & h_{22} & h_{23}  \\ h_{31} & h_{32} & h_{33}  \end{array} \right] \left[ \begin{array}{c} x_1 \\ x_2 \\ x_3 \end{array} \right],
\end{eqnarray}
The system is assumed to be noise free for the moment. The transmitted symbols $x_1,   x_2$,   and $x_3$ are multiplied by factors,   say $v_1,   v_2$,   and $v_3$,  respectively,   and retransmitted over another time slot,   $t^{\prime}$ to assist receivers to cancel interference. The received signal at time slot $t^{\prime}$ would be
\begin{eqnarray}
\left [ \begin{array}{c} y_1^{\prime} \\ y_2^{\prime} \\ y_3^{\prime} \end{array} \right ]= \left [ \begin{array}{c c c} h_{11}^{\prime} & h_{12}^{\prime} & h_{13}^{\prime} \\ h_{21}^{\prime} & h_{22}^{\prime} & h_{23}^{\prime}  \\ h_{31}^{\prime} & h_{32}^{\prime} & h_{33}^{\prime}  \end{array} \right ]\left [ \begin{array}{c} v_1 x_1 \\ v_2 x_2 \\ v_3 x_3 \end{array}\right ].
\end{eqnarray}
 ${\bf y}$ and $ {\bf y}^{\prime}$ are combined linearly to cancel interference. This linear combination has the form of ${\bf U} {\bf y}+{\bf y}^{\prime}$, where  ${\bf U}$ is a $3 \times 3$ diagonal matrix defined as 
\begin{eqnarray}
{\bf U}=\textrm{diag}(u_1, u_2, u_3)
\end{eqnarray}
 Therefore, the following conditions should be satisfied at receiver $1$ for the interference to be cancelled by a linear combination of the received signals.
\begin{eqnarray}
u_1 h_{12}=-v_2  h_{12}^{\prime},   \\
u_1 h_{13}=-v_3  h_{13}^{\prime} 
\label{rec1}
\end{eqnarray}
 Similarly,   the conditions 
\begin{eqnarray}
u_2 h_{21}=- v_1  h_{21}^{\prime},  \\
u_2 h_{23}=- v_3  h_{23}^{\prime}  
\label{ESC}
\end{eqnarray}
 and 
\begin{eqnarray}
u_3 h_{31}=- v_1  h_{31}^{\prime},  \\
u_3 h_{32}=- v_2  h_{32}^{\prime} 
\label{rec3}
\end{eqnarray}
should be met at receivers $2$, and $3$, respectively. Aggregating all IA conditions in (\ref{rec1}),   (\ref{ESC}),   and (\ref{rec3}) into a unified system equation, we obtain
\begin{eqnarray}
{\bf F} {\bf c}={\bf 0},
\label{EC1}
\end{eqnarray}
where ${\bf c}$ is defined as ${\bf c}=\left[ v_1 \quad v_2 \quad v_3 \quad u_1 \quad u_2 \quad u_3 \right ]^T$ and ${\bf F}$ is the unified system matrix defined as
\begin{eqnarray}
{\bf F}= \left[\begin{array}{c c c c c c} h_{21}^{\prime} & 0& 0 & 0 & h_{21} & 0 \\  h_{31}^{\prime} & 0& 0 & 0 & 0& h_{31} \\ 0&  h_{12}^{\prime} & 0& h_{12} & 0 & 0 \\  0 & h_{32}^{\prime} & 0& 0 & 0 &  h_{32}\\  0& 0 & h_{13}^{\prime} & h_{13} & 0 & 0 \\  0& 0 & h_{23}^{\prime} & 0 &  h_{23} & 0 \end{array} \right].
\label{Fmat}
\end{eqnarray}
Note that ${\bf F}$ is a full rank matrix because the channel coefficients are assumed to be generic. Therefore,    (\ref{EC1}) has no nontrivial solution. However,  assume that one of the channel coefficients over time slot $t^{\prime}$,   say $h_{23}^{\prime}$,   has already the proper value to satisfy (\ref{ESC}),   i.e., 
\begin{eqnarray}
h_{23}^{\prime}=-\frac{u_2 h_{23}}{ v_3}.   
\label{BE1}
\end{eqnarray}
Thus,   (\ref{ESC}) is  already satisfied and can be omitted. Therefore,   (\ref{EC1}) is modified as 
\begin{eqnarray}
{\bf F}_n {\bf c}={\bf 0},
\label{EC2}
\end{eqnarray}
where ${\bf F}_n $ is defined as 
\begin{eqnarray}
{\bf F}_r= \left[\begin{array}{c c c c c c} h_{21}^{\prime} & 0& 0 & 0 & h_{21} & 0 \\  h_{31}^{\prime} & 0& 0 & 0 & 0& h_{31} \\ 0&  h_{12}^{\prime} & 0& h_{12} & 0 & 0 \\  0 & h_{32}^{\prime} & 0& 0 & 0 &  h_{32}\\  0& 0 & h_{13}^{\prime} & h_{13} & 0 & 0 \end{array} \right].
\label{Fnmat}
\end{eqnarray}
${\bf F}_r$ is the same with  ${\bf F}$ except that the last row is omitted. Every $ {\bf c} \in \textrm{kernel}({\bf F}_r)$ would satisfy the equation (\ref{EC2}). Since  $\textrm{rank} \left(\text{kernel}({\bf F}_r) \right)=1$, there are  infinite number of solutions for  ${\bf c}$. We can obtain the unique solution by normalizing one of the elements. 

Matrix ${\bf T}$ in (\ref{TM}) can be evaluated as 
\begin{eqnarray}
 &&{\bf T} = ( {\bf H}^{[13]}  ) ^{-1} {\bf H}^{[23]}   ( {\bf H}^{[21]})^{-1} {\bf H}^{[12]} ( {\bf H}^{[32]}  ) ^{-1}  {\bf H}^{[31]}\nonumber\\ &&{=}\: \left[\begin{array}{c c} \frac{h_{23} h_{12} h_{31}}{h_{13}  h_{21}h_{32}} & 0 \\ 0 &  \frac{ h_{23}^{\prime} h_{12}^{\prime} h_{31}^{\prime}}{ h_{13}^{\prime}h_{21}^{\prime} h_{32}^{\prime}}   \end{array} \right ].
\end{eqnarray}
According to (\ref{BE1}) and (\ref{EC2}),   we have 
\begin{eqnarray}
 \frac{h_{23} h_{12} h_{31}}{h_{13}  h_{21}h_{32}}=  \frac{ h_{23}^{\prime} h_{12}^{\prime} h_{31}^{\prime}}{ h_{13}^{\prime}h_{21}^{\prime} h_{32}^{\prime}}.
 \label{examcac}
\end{eqnarray}
Therefore,   ${\bf T}$ becomes 
\begin{eqnarray}
{\bf T} = \left[\begin{array}{c c} \frac{h_{23} h_{12} h_{31}}{h_{13}  h_{21}h_{32}} & 0 \\ 0 &  \frac{h_{23} h_{12} h_{31}}{h_{13}  h_{21}h_{32}}   \end{array} \right ] = {\bf T}_P. 
\end{eqnarray}

This channel aiding condition is the same with the  condition obtained in (\ref{CAEm}). Therefore, we obtained one degree of freedom for each of the users   by using two time slots of the channel.  The total of $3/2$ degrees of freedom is achieved for the entire channel. In comparison to EIA scheme,  \cite{Nazer12},   which requires all elements of the complementary channel matrix ${\bf H}^{\prime}$ to have some specified values; this example just enforces a single condition on channel coefficients,   (\ref{examcac}),   and hence significantly lowers the required delay.

\section{Delay Analysis}

Our scheme relies on matching up channel matrices so that the interference terms cancel out when received signal vectors are combined linearly. Clearly,   given any matrix ${\bf T}$,   the probability that channel aiding condition will occur exactly  is zero (for continuous-valued fading). Thus,   we can only look for channel aiding condition to be satisfied approximately. By taking finer approximations,   we can achieve the target rate in the limit.

Consider the example $3$ user IFC studied in previous section. Perfect IA occurs when the following condition is satisfied over time slots $t$ and $t^{\prime}$,    
\begin{eqnarray}
t^{\prime}&=&(h_{23}^{\prime} h_{12}^{\prime} h_{31}^{\prime})/(h_{13}^{\prime} h_{21}^{\prime} h_{32}^{\prime} )\nonumber \\ &&{=}\:(h_{23} h_{12 }h_{31})/(h_{13} h_{21} h_{32} )=t.
\label{caiac}
\end{eqnarray}
 Based on the assumption that channel coefficients are generic,   the probability of this event is zero and   some sort of approximation have to be used. Following theorem specifies residual interference in terms of the amount of approximation used. 

\begin{theorem}
\label{theores}
If $h_{ij}^{\prime}$  is considered as the complementary channel coefficient satisfied approximately,  residual interference will be 
 \begin{eqnarray}
 d=|h_{ij}^c |^2 |\Delta t/t|^2 |v_j |^2 p, 
 \end{eqnarray}  
where  $h_{ij}^c$ is defined as
\begin{eqnarray}
h_{ij}^c=t/t^{\prime} * h_{ij}^{\prime};
\end{eqnarray}
$\Delta t$ is defined as the amount of approximation used, $t^{\prime}=t+\Delta t$, and $p$ is the average transmitted power from user $j$.
\end{theorem}

\begin{proof}
Let's first evaluate how does $\Delta t$ affects the respective value of $\Delta h$,    assuming $t^{\prime}=t+\Delta t$. Consider the special case of $i=2\; \textrm{and} \; j=3$,  perfect channel coefficient value,    $h^c_{23}$,    would exactly satisfy channel aiding condition,    i.e.,  
\begin{eqnarray}
h_{ij}^c=t/t^{\prime} * h_{ij}^{\prime} \Rightarrow &(h_{23}^{c} h_{12}^{\prime} h_{31}^{\prime})/(h_{13}^{\prime} h_{21}^{\prime} h_{32}^{\prime} )\nonumber \\& =\:(h_{23} h_{12 }h_{31})/(h_{13} h_{21} h_{32} ).
\label{caiace}
\end{eqnarray}
This condition cannot be satisfied exactly and following equation is satisfied on the actual value of $h_{23}^{\prime} $,    
\begin{eqnarray}
&t^{\prime}=t+\Delta t \nonumber \\
&\Rightarrow (h_{23}^{\prime} h_{12}^{\prime} h_{31}^{\prime})/(h_{13}^{\prime} h_{21}^{\prime} h_{32}^{\prime} )\nonumber \\ &{}=\:(h_{23} h_{12 }h_{31})/(h_{13}. h_{21} h_{32} )+\Delta t.
\label{caiaca}
\end{eqnarray}
Combining  (\ref{caiace}) and  (\ref{caiaca}),    we obtain
\begin{eqnarray}
h_{23}^{\prime} =h_{23}^{c}+\frac{h_{23}^{c}}{t} \Delta t.
\label{herr}
\end{eqnarray}
Received signal for second user at time slot $t$ can be evaluated as 
\begin{eqnarray}
y_2=h_{21} x_1+h_{22} x_2+h_{23} x_3.
\end{eqnarray}
And similarly,    received signal at time slot $t^{\prime}$ can be written as
\begin{eqnarray}
y_2^{\prime}=h_{21}^{\prime} v_1 x_1+h_{22}^{\prime} v_2 x_2+h_{23}^{\prime} v_3 x_3.
\end{eqnarray}
We perform the following linear combination on received signals to cancel interference;
\begin{eqnarray}
u_2 y_2+y_2^{\prime}&=&(h_{21} u_2+h_{21}^{\prime} v_1 ) x_1\nonumber \\ &&{+}\:(h_{22} u_2+h_{22}^{\prime} v_2 ) x_2\nonumber \\ &&{+}\:(h_{23} u_2+h_{23}^{\prime} v_3 ) x_3.
\label{lcomb}
\end{eqnarray}
Substituting for $h_{23}^{\prime} $ from (\ref{herr}),   and assuming interference from user $1$ is cancelled out using LIA scheme,   residual interference at user $2$ is obtained as $|h_{23}^c |^2 |\Delta t/t|^2 |v_3 |^2 p$. 

 In general if $h_{ij}^c$  is considered as complementary channel coefficient instead of $h_{23}^c$,    residual interference will be 
 \begin{eqnarray}
 d=|h_{ij}^c |^2 |\Delta t/t|^2 |v_j |^2 p.
 \label{residu}
 \end{eqnarray}
 \end{proof}
 
 Therefore,    to minimize residual interference to transmitted power ratio,    channel aiding condition is considered on an interference with the least value of $h_{ij}^c$,    which is equivalent to the least value of $h_{ij}^{\prime}$  according to (\ref{herr}).  

Now we should evaluate how much of approximation is acceptable for $t^{\prime}$ such that the value of residual interference is not large enough to sacrifice channel degrees of freedom. 
We intend to evaluate how large $\Delta t$ is tolerable to prevent loss in DoF.  Without loss of generality,   it is assumed that  $h_{23}^{\prime}=\min \{h_{ij}^{\prime},   i \neq j\}$. Considering (\ref{herr}) and (\ref{lcomb}),     signal to noise plus interference ratio (SINR) at the receiver $2$ is obtained as 
\begin{eqnarray}
\textrm{SINR}_2=\frac{(h_{22} u_2+h_{22}^{\prime} v_2)^{2} p}{|h_{23}^c/t|^2 |\Delta t|^2 |v_3 |^2 p+N_2 }, 
\label{SINR}
\end{eqnarray}
where $N_2$ is defined as noise variance at receiver $2$ after linear combination,    i.e.,   $N_2=(1+|u_2|^2) \sigma_n^2$.  To achieve optimum DoF,    we should have
\begin{eqnarray}
\lim_{p\to\infty}(|h_{23}^c/t|^2 |\Delta t|^2 |v_3 |^2  p)\leq \infty.
\end{eqnarray}
Noting the plot shown in Fig. \ref{figure:approx} and assuming  $|t| \approx |t^{\prime}|$,    $|\Delta t|^2$ can be approximated as
\begin{eqnarray}
|\Delta t|^2 =4|t |^2 \sin⁡(\Delta\phi/2)^2\leq|t|^2 \Delta\phi^2, 
\label{maxi}
\end{eqnarray}
where $\Delta\phi$ is the maximum tolerable phase difference between $t^{\prime}$ and $t$,    as can be seen in Fig. \ref{figure:approx}.  
 \begin{figure}
\centering \includegraphics[scale=0.55]{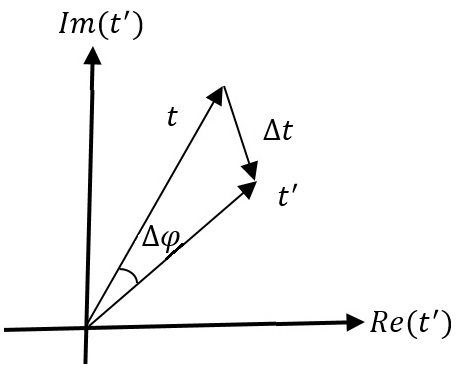}
\caption{Maximum tolerable approximation error between $t^{\prime}$ and $t$. }
\label{figure:approx}
\end{figure}
 Substituting (\ref{maxi}) in (\ref{SINR}),    we find that the SINR per code word symbol is lower bounded as follows
\begin{eqnarray}
\textrm{SINR}_2 \geq \frac{ h_{22} u_2+h_{22}^{\prime} v_2)^{2} p}{|h_{23}^c |^2 \Delta\phi^2 |v_3 |^2 p+N_2 }.
\end{eqnarray}
To maintain a capacity scaling of roughly $1/2\log P$ per user,    we require that,    for some constant $\psi$
\begin{eqnarray}
\Delta\phi^2=(p\psi)^{-1}.
\label{lnda}
\end{eqnarray}
Similar criteria can be obtained on  $\delta=||t^{\prime}|-|t||$ assuming $\Delta\phi \approx 0$.

 \subsection{Expected Delay for Phases}
 
To compute the delay needed for the channel coefficients to match,  we consider the special case of $t$ magnitudes to be fixed. 
If transmitters were to send a new symbol every time slot and the receivers were to simply treat interference as noise,    this expected delay is $1$.  For channel aided alignment,    code word symbols must travel through a channel matrix and the complementary channel matrix with its $t^{\prime}$ (\ref{caiaca}) to have a phase in the interval $[\phi(t)-\Delta\phi,   \phi(t )+\Delta\phi]$. 

 According to the model considered for the channel,    complex random variables $h_{ij}^{\prime}$ have a Gaussian distribution with zero mean and variance $\sigma^2$.  Therefore,    the phase of $t^{\prime}$ would have a uniform distribution in the interval $(-\pi,   +\pi]$,    and the probability that the phase of $t^{\prime}$ lies in the interval $[\phi(t)-\Delta\phi,   \phi(t )+\Delta\phi]$  equals to $1/\eta=2\Delta\phi/2\pi$,    where $\Delta\phi$ is defined as the phase of $\Delta t$,    i.e.,  $\Delta\phi=\phi(\Delta t)$.  Since the channel gains are independent,    the probability of the complementary matrix occurring in a given time slot is $(1/\eta)$.  Thus,    the number of time slots until the complementary matrix occurs is a geometric random variable with parameter $(1/\eta)$ and expected delay is
\begin{eqnarray}
d_{ave}^{CAIA}=\eta=\pi (p\psi)^{(1/2)}.
\end{eqnarray}
For time-varying magnitudes,    the expected delay scales in a similar fashion with an additional penalty for waiting for the magnitudes to match.  In comparison,    the expected delay with  ergodic alignment scheme is \cite{Nazer12}, 
\begin{eqnarray}
d_{ave}^{EIA}=(2p\psi)^{(9/2)}.
\end{eqnarray}
The significant gain in delay performance can be easily seen in the case of $3$ user IFC.  Fig.\ref{figure:pdelay} shows the  simulated results for phase delay performance with $phi(t_{ij})$ values generated using unform distribution in the interval $(-\pi,   +\pi]$. $\phi(\Delta t)$ is assumed to be $\phi(\Delta t) \leq \frac{\pi}{60}$. For ergodic IA scheme,    $\phi(\Delta h)$ is considered as $\phi(\Delta h)=1/\sqrt{2} \phi(\Delta t)$,   to equalize residual interference with both schemes. 
 \begin{figure}
\centering \includegraphics[scale=0.85]{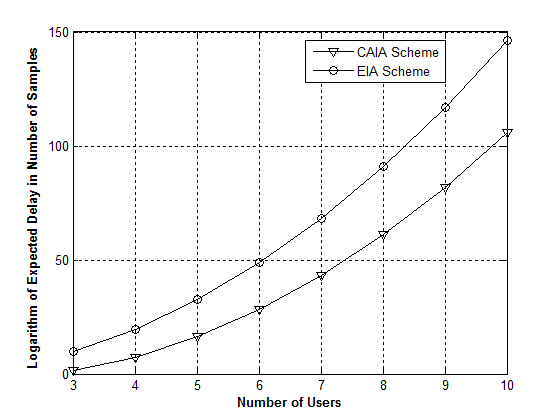}
\caption{Phase Delay Performance vs Number of Users. }
\label{figure:pdelay}
\end{figure}

 \subsection{Expected Delay for the Magnitudes}
 
To compute delay needed for the magnitudes to match,   assuming phase criteria is satisfied already,   i.e.,  $\Delta\phi = 0$  we should evaluate the probability distribution function of the magnitude of $t^{\prime}$ obtained in (\ref{caiac}),    which is not computationally straightforward.  Instead of computing exact distribution function of $r^{\prime}=|t^{\prime}|$,    we would use Gaussian approximation.  Taking the natural logarithm of $r^{\prime}$,    new random variable $z=\ln⁡(r^{\prime} )$ is obtained as
\begin{eqnarray}
z=&\ln⁡(r^{\prime} )=&\ln⁡(r^{\prime}_{23})+\ln⁡(r^{\prime}_{12})+\ln⁡(r^{\prime}_{31})\nonumber\\&&-\ln⁡(r^{\prime}_{13})-\ln⁡(r^{\prime}_{21})-\ln⁡(r^{\prime}_{32}), 
\end{eqnarray}
where $r^{\prime}_{ij},    i,   j=1,   2,   3$  is defined as $r^{\prime}_{ij}=|h_{ij}^{\prime}|$.  Random variable $ z_{ij}=\ln⁡(r^{\prime}_{ij})$ is  Log-Rayleigh distributed with mean value 
\begin{eqnarray}
m_{ij}=\ln⁡(\sigma)+\ln⁡(\sqrt{2})-1/2\lambda, 
\label{meanij}
\end{eqnarray}
where $\lambda$ is the Euler number defined as 
\begin{eqnarray}
\lambda=-\int_{0}^{+\infty}\ln⁡(x)  e^{-x} dx.
\end{eqnarray}
According to Euler-Mascheroni equality,    
\begin{eqnarray}
\lambda^2+\pi^2/6=\int_{0}^{+\infty}\ln⁡(x)^2   e^{-x} dx, 
\end{eqnarray}
therefore, variance of $z_{ij}$ is obtained as 
\begin{eqnarray}
v_{ij}=\frac{\pi^2}{24}.
\label{varij}
\end{eqnarray}
Using (\ref{meanij}) and (\ref{varij}), and based on the Gaussian approximation,  $z$ can be approximated as a Gaussian random variable with zero mean and variance $v_z=\frac{\pi^2}{4}$. Therefore,    $r^{\prime}$ would have the following distribution,   
\begin{eqnarray} 
p(r^{\prime})=\frac{1}{r^{\prime}\sqrt{(2\pi v_z )}} \exp⁡(-\frac{(\ln⁡(r^{\prime}))^2}{2v_z}).
\label{approxdis}
\end{eqnarray}
Simulation results shows the approximation is quite valid. Therefore,    $p_{CA}^r=|r^{\prime}-r|\leq dr$ is obtained as follows
\begin{eqnarray}\begin{split}
&p_{CA}^r (dr|r)=\\&\int_{r-dr}^{r+dr}\frac{1}{r^{\prime}\sqrt{(2\pi v_z )}} \exp⁡(-\frac{(\ln⁡(r^{\prime}))^2}{2v_z})dr^{\prime}=\\& Q(\frac{\ln⁡(r-dr)}{\sqrt{v_z}})-Q(\frac{\ln⁡(r+dr)}{\sqrt{v_z}}), 
\end{split}\end{eqnarray}
where $Q$ is defined as complementary error function. 
\begin{eqnarray}
Q(x)=\int_{x}^{\infty} \frac{1}{\sqrt{2\pi}} e^{-\frac{r^2}{2}}dr.
\end{eqnarray}

The expected delay for the magnitudes to match is obtained as $d(dr|r)=\frac{1}{p_{CA}^r (dr|r)}$.  In contrast,    the expected delay for ergodic alignment scheme is obtained as $d_{EIA}=\frac{1}{\prod(p_{EAI}^r(r_{ij}))}$, where $p_{EAI}^r(r_{ij})$ is obtained as
\begin{eqnarray} \begin{split}
p_{EAI}^r(r_{ij})=&p_E^r (|r_{ij}^{\prime}-r_{ij} |\leq \Delta h|r_{ij})=\\&\int_{r_ij-\Delta h}^{r_ij+\Delta h} \frac{r_{ij}^{\prime}}{\sigma^2}  \exp⁡(-\frac{(r_{ij}^{\prime})^2}{2\sigma^2 })dr_{ij}^{\prime}=\\ &\exp(-\frac{(r_{ij}-\Delta h)^2}{2\sigma^2 })-\exp(-\frac{(r_{ij}+\Delta h)^2}{2\sigma^2 }).
\end{split}\end{eqnarray}
Fig.\ref{figure:delay} shows the numerical results with $r_{ij}$ values generated using Rayleigh distribution with parameter $\sigma^2=1$.  For ergodic IA scheme,    $\Delta h$ is considered as $\Delta h=1/\sqrt{2} |\Delta t/t| min(|h_{ij}|)$,   to equalize residual interference with both schemes.  The main difference in delay performance between these two schemes is merely due to the number of constraints and not the nature of the constraints themselves. For the general case of $K$ user IFC,   number of constraints is $K^2$ for ergodic alignment scheme,    while in channel aided alignment scheme,    this number is $(K-1)(K-2)-1$.  Delay performance for magnitudes to match is not straightforward to evaluate for the case of more than $3$ users because $r_{ij}$ values are dependent on each other. However, it is fair to argue that this delay is upper bounded by that of the EIA scheme.

  \begin{figure}
\centering \includegraphics[scale=0.6]{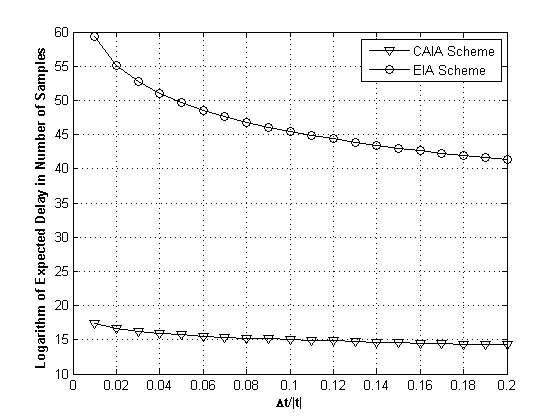}
\caption{Magnitude matching delay performance,   compared between EIA and CAIA schemes. }
\label{figure:delay}
\end{figure}

\section {Conclusion}

The concept of channel aiding conditions are introduced in this paper. These conditions are equivalent to perfect IA feasibility conditions on channel structure. Channel aided IA  scheme was proposed based on these conditions to achieve optimal degrees of freedom of the IFC using  limited number of channel extension. This scheme significantly lowers expected delay in comparison to  previously proposed methods. This approach makes the best use of the possible linear structure of the channel.

Assuming generic channel coefficients, stated conditions on channel structure are not exactly feasible. Approximation should be used and its effect on residual interference have to be analyzed. This is the subject of our future work on CAIA. Overall,   the proposed method is capable of reducing dimensionality and signal to noise ratio needed to exploit DoF benefits of IA schemes.

\bibliographystyle{IEEEtran}

\end{document}